\documentclass[journal,twoside,web]{ieeecolor}

\usepackage{lcsys}

\usepackage{amsmath,amssymb,amsfonts}
\usepackage{braket}
\allowdisplaybreaks

\usepackage{graphicx}
\usepackage{placeins}
\usepackage{algorithm}
\usepackage[noend]{algpseudocode}
\usepackage{flushend}
\usepackage{booktabs}
\usepackage{caption}
\usepackage{hyperref}
\hypersetup{
  colorlinks = true,
  linkcolor  = blue,
  citecolor  = blue,
  urlcolor   = cyan,
  filecolor  = magenta,
}

\newtheorem{theorem}{Theorem}
\newtheorem{corollary}{Corollary}
\newtheorem{lemma}{Lemma}
\newtheorem{remark}{Remark}

\DeclareMathOperator{\Tr}{Tr}
\newcommand{\e}{\mathrm{e}}

\title{Fidelity-Based Robustness Margins \\ for Quantum Control}

\author{S.\,P.\ O'Neil$^{1,*}$,
  F.\,C.\ Langbein$^2$,
  C.\,A.\ Weidner$^3$,
  E.\,A.\ Jonckheere$^4$,
  and S.\ Schirmer$^5$
  \thanks{Any opinions in this work are solely those of the authors and do not reflect those of the U.S.\ Army, the USMA, or the DoD\@.}
  \thanks{$^1$ Department of Electrical Engineering \& Computer Science, United States Military Academy, NY, USA\@.
  \texttt{sean.oneil@westpoint.edu}}
  \thanks{$^2$ School of Computer Science and Informatics, Cardiff University, Cardiff, CF24 4AG, UK\@.
  \texttt{frank@langbein.org}}
  \thanks{$^3$ Quantum Engineering Technology Laboratories, H.\,H.\ Wills Physics Laboratory and Department of Electrical and Electronic Engineering, University of Bristol, Bristol BS8 1FD, UK\@.
  \texttt{c.weidner@bristol.ac.uk}}
  \thanks{$^4$ Dept of Electrical \& Computer Engineering, University of Southern California, CA, USA\@.
  \texttt{jonckhee@usc.edu}}
  \thanks{$^5$ Faculty of Science \& Engineering, Physics, Swansea University, UK\@.
  \texttt{s.m.shermer@gmail.com}}
}

\newcommand{\QRMAuditRays}{366}
\newcommand{\QRMAuditKMax}{10000}
\newcommand{\QRMAuditGuardRays}{0}

\newcommand{\QRMAuditContinuationMedian}{73}

\newcommand{\QRMAuditRefinementMedian}{319}

\newcommand{\QRMAuditControllers}{61}
\newcommand{\QRMAuditDriftSpread}{1.6}
\newcommand{\QRMAuditControlOneSpread}{2.7}
\newcommand{\QRMAuditControlTwoSpread}{2.9}
\newcommand{\QRMAuditRelativeWidth}{5.5\times10^{-4}}
\newcommand{\QRMAuditTolerance}{1\times10^{-10}}

\begin{document}

\maketitle
\thispagestyle{empty}

\begin{abstract}
  We develop certified lower bounds on structure-specific fidelity-threshold robustness margins for finite-dimensional closed quantum systems under piecewise-constant coherent control. A scalar physical parameter may perturb the drift, a control Hamiltonian, or another declared Hamiltonian component across the control horizon. A differential sensitivity bound for the trace-amplitude gate fidelity yields a threshold-dependent Lipschitz constant on the nominal connected safe parameter component and hence a certified finite perturbation radius. Recentering this certificate produces an iterative one-dimensional method whose finite runs return lower bounds on the distance to the first threshold contact. The fidelity may touch the threshold without falling below it. A three-qubit gate-control example shows that the returned bounds vary by up to a factor of three across controllers of comparable nominal fidelity and contain structure-dependent information not captured by nominal differential sensitivity alone.
\end{abstract}

\section{Introduction}

\IEEEPARstart{Q}{uantum} technologies require precise control of system dynamics, but model perturbations and environmental coupling can substantially reduce performance and eliminate quantum advantage~\cite{Koch_2022}. Robustness margins are therefore important~\cite{Petersen2013,Automatica}. Classical tools based on stability margins or the structured singular value ($\mu$)~\cite{Doyle1982},~\cite[Chap.~10]{Zhou} are of limited utility under coherent control: the ideal closed dynamics are typically marginally stable, and the objective is a finite-time fidelity at a prescribed readout time $t_f$ rather than asymptotic stability. Earlier work derived conservative performance guarantees based on the norm of the sensitivity and iterative searches over locally worst-case uncertainty directions~\cite{schirmer2024,oneil_2024_sensitivity_bounds}. Finite-time guarantees can also be derived from norm- or set-based uncertainty descriptions: Lidar, Zanardi, and Khodjasteh bound distances between quantum evolutions by Hamiltonian operator norms~\cite{lidar2008}, while Berberich, Fink, and Holm bound worst-case fidelity for circuits with independent multiplicative errors in the complete gate generators~\cite{berberich2024}, later extended to a set-based framework covering coherent, time-dependent, and Markovian errors~\cite{berberich2025}. These methods give global, non-iterative guarantees from declared gate-level error sets and generator information. Kosut, Lidar, and Rabitz provide a broader universal time-bandwidth bound over general uncertainty classes~\cite{kosut2025} for perfect (unit-fidelity) nominal controllers, using aggregate norm and interaction-picture uncertainty measures rather than an explicit physical parameterization of the perturbation. Other approaches focus on robust controller synthesis~\cite{kiely2024,zou2025}.

We consider \emph{post-design certification} for unitary gate control with a given, possibly \emph{imperfect}, nominal controller. The \emph{structured} affine Hamiltonian perturbation affects the drift, a control Hamiltonian, or selected Hermitian matrix entries, with known structure. Piecewise-constant control permits interval-wise bounds on the fidelity derivative~\cite{oneil_2024_sensitivity_bounds}. We convert these bounds into finite threshold certificates and recenter them toward the first boundary of the nominal connected safe component, one scalar parameter at a time. These margins are analogous \emph{in purpose} to classical structured robustness margins, but certify finite-time fidelity rather than frequency-domain $\mu$-bounds~\cite{Doyle1982},~\cite[Chap.~10]{Zhou}.

Our contributions are: (i)~a fidelity-dependent bound on the sensitivity of the trace-amplitude gate-fidelity for a declared Hermitian perturbation structure; (ii)~a threshold-dependent local Lipschitz radius on the nominal connected safe component, allowing imperfect nominal target fidelity; (iii)~a certified adaptive continuation algorithm for a fixed physical parameter that advances toward the first threshold boundary in each direction; and (iv)~a three-qubit case study showing that, particularly for control-Hamiltonian uncertainty, the computed lower bounds vary substantially even when nominal fidelity and nominal differential sensitivity do not provide a consistent controller ranking.

\section{Quantum dynamics and control problem\label{sec:quantum control problem}}

We consider a closed quantum system with Hilbert space $\mathcal{H}$ of dimension $N<\infty$, so pure states are complex vectors $\psi(t) \in \mathbb{C}^N$ in an orthonormal basis of $\mathcal{H}$, evolving by the Schr\"odinger equation $\mathrm{i}\hbar \tfrac{d}{dt} \psi(t) = H(t) \psi(t)$ with self-adjoint $H(t)$. We introduce $M$ control fields steering the dynamics over a readout interval $[0,t_f]$, restricted, as usual, to piecewise-constant functions on $\tau$ uniform intervals of length $\Delta$, with $t_k=k \Delta$, $t_0=0$ and $t_f=\tau \Delta$~\cite{Koch_2022,KHANEJA_2005}. Each control pulse $f_{m}^{(k)} \in \mathbb{R}$ enters through an interaction Hamiltonian $H_m$ such that in the $k$-th interval
\begin{equation}
  H^{(k)} = H_0 + \sum\limits_{m = 1}^M H_m f_{m}^{(k)}, \quad 1 \leq k \leq \tau,
\end{equation}
where $H_0$ is the drift Hamiltonian. With $\hbar=1$, the solution on $[t_{k-1},t_k]$ is $U^{(k)}  = \exp[-\mathrm{i}H^{(k)} \Delta]$, so the evolution $\dot{U}(t) = -\mathrm{i} H(t) U(t)$ with $U(0) = I$ gives
\begin{equation}\label{eq:schrodinger_gate}
  U(t_f) = \prod_{k=1}^{\tau}U^{(k)} = U^{(\tau)} U^{(\tau-1)} \cdots U^{(1)},
\end{equation}
where $\prod_{k=1}^{\tau}$ denotes an ordered product. We optimize $U(t_f) \in \mathbb{U}(N)$ to maximize its overlap with a target gate $U_f \in \mathbb{U}(N)$, measured by the normalized trace-amplitude gate fidelity at $t_f$,
\begin{equation}\label{eq:fidelity}
  \mathcal{F}(t_f) = \frac{1}{N}\left| \Tr \left[ U_f^{\dagger} U(t_f) \right] \right| .
\end{equation}
The fidelity error is $\varepsilon(t_f) = 1 - \mathcal{F}(t_f)$. Analogous state-transfer fidelities can be defined, but we focus on gate fidelity.

\section{Structured Perturbations \label{sec:uncertainty_model}}

Consider a scalar uncertain parameter $\mu\in\mathbb R$ with nominal value $\mu_0$. Deviations from the nominal Hamiltonian are modeled by the structured term $(\mu-\mu_0)\hat{H}_\mu\alpha_{\hat{H}}^{(k)}$, where $\hat{H}_\mu$ is the Hermitian perturbation structure and $\alpha_{\hat{H}}^{(k)}\in\mathbb{R}$ its interval-dependent strength. The parameter $\mu$ is assumed to lie in some set $\Omega_\mu = [\underline{\mu},\overline{\mu}]$ with $\mu_0 \in \Omega_\mu$. The uncertain Hamiltonian in time-step $k$ is then
\begin{equation}\label{eq:perturbed_hamiltonian}
  \tilde{H}^{(k)}  = H_0 + \sum_{m=1}^{M} H_m f_{m}^{(k)} + \left( \mu-\mu_0 \right) \hat{H}_\mu\alpha_{\hat{H}}^{(k)}.
\end{equation}
The strength $\alpha_{\hat{H}}^{(k)}$ lets the structure enter with interval-dependent weight. For collective uncertainty in the drift, $\hat{H}_\mu = H_0$ and $\alpha_{\hat{H}}^{(k)} = 1$; for multiplicative uncertainty in an interaction Hamiltonian $H_m$, take $\hat{H}_\mu = H_m$ and $\alpha_{\hat{H}}^{(k)} = f_{m}^{(k)}$. Collective uncertainty is chosen for ease of illustration; the analysis applies equally to a parameter affecting only select entries, $\tilde{H} = H_{\mathrm{nom}} + \delta\hat{H}_\mu$ with $H_{\mathrm{nom}} \neq \hat{H}_\mu$ and $\hat{H}_\mu$ sparse.

The perturbed solution to~\eqref{eq:schrodinger_gate} at $t_f$ is given by
\begin{subequations}
  \begin{align}
    \tilde{U}(t_f) &= \prod_{k=1}^{\tau}\tilde{U}^{(k)},\label{eq:perturbed_phi}\\
    \tilde{U}^{(k)} &= \exp\left[-\mathrm{i}\tilde{H}^{(k)} \Delta \right].\label{eq:perturbed_Phi2}
  \end{align}
\end{subequations}
The perturbed fidelity at parameter value $\mu$ is
\begin{equation}\label{eq:perturbed_fidelity}
  {\mathcal{F}}_\mu(t_f) = \frac{1}{N} \left| \Tr\left[U_f^\dagger \tilde{U}(t_f) \right] \right|.
\end{equation}
We drop the explicit $t_f$ below; the fidelity is always measured at the readout time.

The same scalar parameter is shared across all control intervals, so it can represent a correlated physical parameter or calibration error. This differs from circuit-level models in which separate multiplicative errors scale the complete generators of individual gates~\cite{berberich2024}, and permits component-wise or sparse structures rather than scaling the whole interval Hamiltonian.

We consider scalar structured Hamiltonian perturbations, one parameter at a time. The analysis excludes simultaneous multi-parameter uncertainty, open-system (Lindblad) perturbations, and stochastic noise.

\section{Sensitivity of the Fidelity\label{sec:fidelity-sensitivity}}
Theorem~\ref{thm:sens_bound} bounds fidelity sensitivity for the certificates in Section~\ref{sec:robustness-margins}. Two preliminary lemmas establish skew-Hermitian derivatives and a fidelity-dependent norm estimate. We assume $\Tr(U_f^\dagger U(t_f))\neq 0$, equivalently $\mathcal{F}_{\mu_0}>0$, and define the nominal phase
\begin{equation}\label{eq:phase}
  \e^{\mathrm{i} \varphi} =
  \Tr\!\left(U_f^\dagger U(t_f)\right)
  \Big/ \bigl|\Tr\!\left(U_f^\dagger U(t_f)\right)\bigr|.
\end{equation}
Writing $z(\mu):=\Tr[U_f^\dagger \tilde{U}(t_f)]$, the derivative of the modulus at the nominal point is
\[
  \left.\frac{d}{d\mu}|z(\mu)|\right|_{\mu_0}
  = \Re\!\left[ \e^{-\mathrm{i}\varphi} \left.\frac{dz}{d\mu}\right|_{\mu_0} \right].
\]
The same argument applies at any parameter value with non-zero overlap by re-centering the phase at that value. The derivative of the perturbed unitary is computed as usual~\cite{oneil_2024_sensitivity_bounds}
\begin{equation}\label{eq:dU}
  \left. \frac{\partial \tilde{U}^{(k)} }{\partial \mu} \right|_{\mu_0}
  = -\mathrm{i} \Delta \int_{0}^1 e^{ -\mathrm{i} \Delta H^{(k)}(1-s)}
  \left. \frac{\partial \tilde{H}^{(k)}}{\partial \mu} \right|_{\mu_0}
  e^{-\mathrm{i} \Delta H^{(k)}s}\, ds,
\end{equation}
and we define
\begin{subequations}
  \begin{align}
    Z^{(k)} &:= U^{(k-1)} \hdots U^{(1)} U_f^\dagger U^{(\tau)} \hdots U^{(k+1)} U^{(k)} e^{-\mathrm{i} \varphi}, \\
    X_{\mu_0}^{(k)} &:={U^{(k)}}^\dagger \left. \frac{\partial \tilde{U}^{(k)}}{\partial \mu} \right|_{\mu_0} \nonumber\\
    &= -\mathrm{i} \Delta \int_{0}^1 e^{ \mathrm{i} \Delta H^{(k)}s}
    \left. \frac{\partial \tilde{H}^{(k)}}{\partial \mu} \right|_{\mu_0}
    e^{-\mathrm{i} \Delta H^{(k)}s}\, ds.\label{eq:Xb}
  \end{align}
\end{subequations}
Following earlier work~\cite{oneil_2024_sensitivity_bounds}, the sensitivity of the fidelity about the nominal value $\mu_0$ is
\begin{equation}\label{eq:sens}
  \zeta_{\mu_0} := \left. \frac{\partial {\mathcal{F}}_\mu}{\partial \mu} \right|_{\mu_0}
  = \frac{1}{N} \sum_{k=1}^{\tau} \Re\Tr\!\left[ Z^{(k)}X_{\mu_0}^{(k)} \right].
\end{equation}
Equivalently, writing
\begin{equation}\label{eq:product_derivative}
  D_{\mu_0}^{(k)}:=U^{(\tau)}U^{(\tau-1)} \hdots
  \left. \frac{\partial \tilde{U}^{(k)}}{\partial \mu} \right|_{\mu_0}
  \hdots U^{(1)},
\end{equation}
one has $\zeta_{\mu_0}=\frac{1}{N}\sum_{k=1}^{\tau}\Re\Tr\!\bigl(U_f^\dagger D_{\mu_0}^{(k)} e^{-\mathrm{i}\varphi}\bigr)$.
The equivalence between the $D_{\mu_0}^{(k)}$ and $Z^{(k)}X_{\mu_0}^{(k)}$ forms follows from
\[
  \left. \frac{\partial \tilde{U}^{(k)}}{\partial\mu} \right|_{\mu_0} = U^{(k)}X_{\mu_0}^{(k)}
\]
and cyclicity of the trace.

\begin{lemma}\label{lemma:X_in_uN}
  Suppose $\tilde{H}^{(k)}(\mu)$ is Hermitian for real $\mu$. Then $X_{\mu_0}^{(k)}\in\mathfrak{u}(N)$, i.e.\ $X_{\mu_0}^{(k)}$ is skew-Hermitian.
\end{lemma}

\begin{proof}
  Directly from~\eqref{eq:Xb}, for each $s\in[0,1]$ the integrand $e^{\mathrm{i}\Delta H^{(k)}s}\bigl(\partial\tilde{H}^{(k)}/\partial\mu|_{\mu_0}\bigr)e^{-\mathrm{i}\Delta H^{(k)}s}$ is Hermitian (unitary conjugation of a Hermitian matrix). Hence, $X_{\mu_0}^{(k)} = -\mathrm{i}\Delta$ times the integral of a Hermitian matrix, which is skew-Hermitian, i.e.\ $X_{\mu_0}^{(k)}\in\mathfrak{u}(N)$.
\end{proof}

\begin{lemma}\label{lemma:ZH_ZSH}
  For $Z^{(k)}=Z_H^{(k)}+Z_{SH}^{(k)}$, with
  \[
    Z_H^{(k)}=\frac{Z^{(k)}+{Z^{(k)}}^\dagger}{2},\qquad
    Z_{SH}^{(k)}=\frac{Z^{(k)}-{Z^{(k)}}^\dagger}{2},
  \]
  we have $\Tr \left( Z_H^{(k)} \right) =N\mathcal{F}$ and $\|Z_{SH}^{(k)}\|_F^2\le N(1-\mathcal{F}^2)$.
\end{lemma}

\begin{proof}
  Let $\mathcal{F} = \mathcal{F}_{\mu_0}$ be the nominal fidelity. 
  By cyclicity,
  \[
    \Tr \left( Z^{(k)} \right) =
    \Tr\!\left( U_f^\dagger U(t_f)e^{-\mathrm{i}\varphi} \right) =
    \left| \Tr\!\left( U_f^\dagger U(t_f) \right) \right| =
    N\mathcal{F}.
  \]
  Thus, $\Tr \left( Z^{(k)} \right)$ is real. Since $\Tr \left( Z_{SH}^{(k)} \right)$ is purely imaginary, $\Tr \left( Z_H^{(k)} \right) =N\mathcal{F}$. Also, $Z^{(k)}$ is unitary, so $\|Z^{(k)}\|_F^2=N$, and the Hermitian and skew-Hermitian parts are orthogonal with respect to the real Hilbert--Schmidt inner product. Hence,
  \[
    \|Z_{SH}^{(k)}\|_F^2 = N-\|Z_H^{(k)}\|_F^2.
  \]
  Applying Cauchy--Schwarz to $I$ and $Z_H^{(k)}$ gives $|\Tr Z_H^{(k)}|^2\le N\|Z_H^{(k)}\|_F^2$, so $\|Z_H^{(k)}\|_F^2\ge N\mathcal{F}^2$, and $\|Z_{SH}^{(k)}\|_F^2\le N(1-\mathcal{F}^2)$.
\end{proof}

\begin{theorem}\label{thm:sens_bound}
  The size of the sensitivity $|\zeta_{\mu_0}|$ is bounded above by
  \[
    | \zeta_{\mu_0} | \leq \sqrt{\frac{1 - \mathcal{F}^2 }{N}} \sum_{k = 1}^\tau \left\| \left. \frac{\partial \tilde{U}^{(k)}}{\partial \mu} \right|_{\mu_0} \right\|_{F}.
  \]
\end{theorem}

\begin{proof}
  By~\eqref{eq:sens}, Lemma~\ref{lemma:X_in_uN}, and the decomposition $Z^{(k)}=Z_H^{(k)}+Z_{SH}^{(k)}$,
  \[
    \begin{aligned}
      |\zeta_{\mu_0}|
      &\le \frac{1}{N} \sum_{k=1}^{\tau} \left| \Re\Tr\!\left[ Z^{(k)}X_{\mu_0}^{(k)} \right] \right|
      = \frac{1}{N} \sum_{k=1}^{\tau} \left| \Tr\!\left[ Z_{SH}^{(k)}X_{\mu_0}^{(k)} \right] \right| \\
      &\le \frac{1}{N} \sum_{k=1}^{\tau} \|Z_{SH}^{(k)}\|_F \|X_{\mu_0}^{(k)}\|_F \\
      &\le \sqrt{\frac{1-\mathcal{F}^2}{N}} \sum_{k=1}^{\tau} \left\| \left.
      \frac{\partial\tilde{U}^{(k)}}{\partial\mu} \right|_{\mu_0} \right\|_F .
    \end{aligned}
  \]
  The equality uses that $Z_H^{(k)}$ is Hermitian and $X_{\mu_0}^{(k)}$ is skew-Hermitian, so $\Tr \left(Z_H^{(k)}X_{\mu_0}^{(k)} \right) \in i \mathbb{R}$ and the product of two skew-Hermitian factors has real trace. Lemma~\ref{lemma:ZH_ZSH} and unitary invariance of the Frobenius norm justify the last two steps.
\end{proof}

\begin{corollary}
  For any scalar structured Hamiltonian perturbation of the form considered above with bounded Frobenius norm, the absolute sensitivity $|\zeta_{\mu_0}|$ vanishes when the nominal fidelity $\mathcal{F} = 1$.
\end{corollary}
This extends~\cite[Lemma~3]{oneil_geometric}, proved for static controls, to piecewise-constant controls and the fidelity-dependent bound improves the bounds in \cite{oneil_2024_sensitivity_bounds}, in a way that is crucial for the certified adaptive continuation algorithm in the next section.

\begin{remark}
  For $\varepsilon=1-\mathcal{F}>0$, the logarithmic sensitivity satisfies
  \[
    \left| \frac{\partial \log \varepsilon}{\partial \mu} \right|
    = \frac{|\zeta_{\mu_0}|}{\varepsilon}.
  \]
  Near perfect fidelity, $\sqrt{1-\mathcal{F}^2}\approx\sqrt{2\varepsilon}$, so the corresponding bound scales as $O(\varepsilon^{-1/2})$. This is consistent with the divergence of the log-sensitivity of the fidelity error for static control fields observed in~\cite{oneil_2024_log_sens}.
\end{remark}

\section{Robustness Margins\label{sec:robustness-margins}}
The sensitivity bound in Theorem~\ref{thm:sens_bound} certifies finite deviations of $\mu$ satisfying $\mathcal{F}_\mu\ge\mathcal{F}_{\text{Th}}$, with strict inequality inside the connected safe component. We construct a uniform sensitivity bound on this component, a Lipschitz radius, and an iterative one-dimensional margin computation.

\paragraph{Safe set}
Assume $0<\mathcal{F}_{\text{Th}}<\mathcal{F}_{\mu_0}$, so only controllers meeting the minimum nominal performance are considered; on $\{\mu:\mathcal{F}_\mu>\mathcal{F}_\text{Th}\}$ the absolute-value fidelity is then differentiable. Continuity of $\mathcal{F}_\mu$ in $\mu$ follows as in~\cite{oneil_2024_sensitivity_bounds}, from continuity of each $\tilde{H}^{(k)}$ together with that of the exponential map, matrix multiplication, the trace, and the absolute value. Define the safe set $\mathcal{I}$ as the connected component containing $\mu_0$ of $\{\mu \in \Omega_\mu : \mathcal{F}_\mu > \mathcal{F}_\text{Th}\}$, a relatively open subinterval of $\Omega_\mu$.

\paragraph{Uniform sensitivity bound on $\mathcal{I}$}
Define $c_\mu:=N^{-1}\Tr(\hat{H}_\mu)$ and the centered structure $\overline{H}_\mu:=\hat{H}_\mu-c_\mu I$. Replacing $\hat{H}_\mu$ by $\overline{H}_\mu$ in~\eqref{eq:perturbed_hamiltonian} gives the interval propagators
  \[
    \overline{U}^{(k)}=e^{\mathrm{i}(\mu-\mu_0)\Delta c_\mu\alpha_{\hat{H}}^{(k)}}\tilde{U}^{(k)}.
  \]
Their product differs from the original evolution only by a global phase, so the fidelity $\mathcal{F}_\mu$ is unchanged. Apply Theorem~\ref{thm:sens_bound} to these propagators and define $B_T:=\sqrt{(1-\mathcal{F}_\text{Th}^2)/N}$ and $g(\mu):=\sum_{k=1}^\tau\|\partial_\mu\overline{U}^{(k)}\|_F$. Since $\mathcal{F}_\mu>\mathcal{F}_\text{Th}$ on $\mathcal{I}$,
\begin{equation}\label{eq:bound_2}
  |\zeta_\mu|\le B_T g(\mu),\qquad \mu\in\mathcal{I}.
\end{equation}
The integral representation~\eqref{eq:dU} permits evaluation of $g(\mu)$, but we seek an inexpensive uniform bound that avoids recomputing these derivatives along the parameter path.

\begin{lemma}\label{lemma:dU_bound}
  We have $\left\| \frac{\partial \tilde{U}^{(k)}}{\partial \mu}\right\|_F \leq \Delta \left\|  \frac{\partial \tilde{H}^{(k)}}{\partial \mu} \right\|_F$.
\end{lemma}

\begin{proof}
  The integral representation for $\partial\tilde{U}^{(k)}/\partial\mu$ is~\eqref{eq:dU} with $\tilde{H}^{(k)}$ in place of $H^{(k)}$. Its exponential factors are unitary for all $\mu\in\mathcal{I}$, so the claim follows from the triangle inequality and unitary invariance of the Frobenius norm.
\end{proof}
Applying Lemma~\ref{lemma:dU_bound} to the centered dynamics gives $g(\mu)\le C_{\hat{H}}$, where $C_{\hat{H}} = \Delta\sum_{k=1}^{\tau} \bigl|\alpha_{\hat{H}}^{(k)}\bigr|\,\|\overline{H}_\mu\|_F$. The structure $\hat{H}_\mu$ and weights $\alpha_{\hat{H}}^{(k)}$ are fixed with respect to $\mu$ in the affine model, so $C_{\hat{H}}$ is independent of the parameter value. Centering cannot increase this bound, since $\|\overline{H}_\mu\|_F\le\|\hat{H}_\mu\|_F$. If $C_{\hat{H}}=0$, then $\mathcal{F}_\mu=\mathcal{F}_{\mu_0}$ throughout $\Omega_\mu$ and the margin is limited only by the parameter domain; henceforth we assume $C_{\hat{H}}>0$. Writing $\overline{H}_m = H_m - N^{-1}(\Tr H_m) I$, this gives $C_{\hat{H}}=t_f\|\overline{H}_0\|_F$ for drift uncertainty and $C_{\hat{H}}=\Delta\|f_m\|_{\ell^1}\|\overline{H}_m\|_F$ for uncertainty in the $m$th interaction Hamiltonian. The case-study structures $H_0$, $H_1$, $H_2$ in Section~\ref{sec:example} are traceless.

\paragraph{Lipschitz bound and certified safe radius}
Although $\mathcal{F}_\mu$ is real analytic away from zero overlap, the certificate requires only a uniform Lipschitz bound on the connected safe component.
\begin{lemma}\label{lemma:lipschitz}
  For any $\mu_a,\mu_b\in\mathcal{I}$,
  \[
    |\mathcal{F}_{\mu_b}-\mathcal{F}_{\mu_a}| \le L_{\hat{H}}|\mu_b-\mu_a|,
  \]
  where $L_{\hat{H}}:=B_T C_{\hat{H}}$.
\end{lemma}

\begin{proof}
  Since $\mathcal{I}$ is a connected interval, the segment between $\mu_a$ and $\mu_b$ lies in $\mathcal{I}$, so
  \[
    |\mathcal{F}_{\mu_b}-\mathcal{F}_{\mu_a}| \le
    \int_{\min(\mu_a,\mu_b)}^{\max(\mu_a,\mu_b)} |\zeta_\xi|\,d\xi .
  \]
  By~\eqref{eq:bound_2} and the centered bound $g(\xi)\le C_{\hat{H}}$, $|\zeta_\xi|\le B_T C_{\hat{H}}$ on $\mathcal{I}$. Hence,
  \[
    |\mathcal{F}_{\mu_b}-\mathcal{F}_{\mu_a}|
    \le B_T C_{\hat{H}}|\mu_b-\mu_a|
    = L_{\hat{H}}|\mu_b-\mu_a|.
  \]
\end{proof}

This gives an explicit safety certificate: parameter values strictly inside the radius remain in the connected safe component, and by continuity the threshold condition still holds on its boundary.

\begin{theorem}\label{thm:safe_radius}
  Let $\nu\in\mathcal{I}$ and define
  \( r_\nu=\frac{\mathcal{F}_{\nu}-\mathcal{F}_\text{Th}}{L_{\hat{H}}}. \)
  For any $\mu_1\in\Omega_\mu$ satisfying
  \( |\mu_1-\nu|\le r_\nu, \)
  the performance criterion $\mathcal{F}_{\mu_1}\ge\mathcal{F}_\text{Th}$ holds. Moreover, if $|\mu_1-\nu|<r_\nu$, then $\mu_1\in\mathcal{I}$, and hence $\mathcal{F}_{\mu_1}>\mathcal{F}_\text{Th}$.
\end{theorem}
\begin{proof}
  First suppose $|\mu_1-\nu|<r_\nu$, and let $ \gamma(\lambda)=\nu+\lambda(\mu_1-\nu), \lambda\in[0,1]$. Starting from $\gamma(0)=\nu\in\mathcal{I}$, define
  \[ \lambda^\star=\sup\bigl\{\lambda\in[0,1]:\gamma([0,\lambda))\subset\mathcal{I}\bigr\}. \]
  For every $\lambda<\lambda^\star$ the segment $\gamma([0,\lambda])$ lies in $\mathcal{I}$, so the absolute-value bound in Lemma~\ref{lemma:lipschitz} implies
  \[ \mathcal{F}_{\gamma(\lambda)}
    \ge \mathcal{F}_{\nu} - L_{\hat{H}}\lambda|\mu_1-\nu|
  > \mathcal{F}_{\nu} - L_{\hat{H}}r_\nu = \mathcal{F}_\text{Th}, \]
  using $|\mu_1-\nu|<r_\nu$. To prove that $\lambda^*=1$, assume that $\lambda^\star<1$; then by continuity of $\mathcal{F}_\mu$ the same estimate passes to the limit,
  \[ \mathcal{F}_{\gamma(\lambda^\star)}
    \ge \mathcal{F}_{\nu} - L_{\hat{H}}\lambda^\star|\mu_1-\nu|
  > \mathcal{F}_\text{Th}. \]
  Because $\mathcal{F}_{\gamma(\lambda^\star)}>\mathcal{F}_\text{Th}$, the connected set $\gamma([0,\lambda^\star])$ lies in the safe set and intersects $\mathcal{I}$; hence it is contained in $\mathcal{I}$. Since $\mathcal{I}$ is relatively open, $\gamma(\lambda)\in\mathcal{I}$ for some $\lambda>\lambda^\star$, contradicting the definition of $\lambda^\star$. Hence $\lambda^\star=1$. By continuity as $\lambda\uparrow1$, $\mathcal{F}_{\gamma(1)}\ge\mathcal{F}_{\nu}-L_{\hat{H}}|\mu_1-\nu|>\mathcal{F}_\text{Th}$. Thus $\gamma([0,1])$ is a connected subset of the strict safe set containing $\nu$, so $\mu_1=\gamma(1)\in\mathcal{I}$.

  If $|\mu_1-\nu|=r_\nu$, then $\mu_1$ may lie on $\partial\mathcal{I}$, where Lemma~\ref{lemma:lipschitz} does not apply; take points on the same segment converging to $\mu_1$ from strict distance less than $r_\nu$. The strict case and continuity of $\mathcal{F}_\mu$ give $\mathcal{F}_{\mu_1}\ge\mathcal{F}_\text{Th}$.
\end{proof}

\begin{algorithm}[!t]
  \caption{Calculation of Robustness Margin $\mathcal{M}$}\label{algorithm}
  \algnewcommand{\LineComment}[1]{\State $\triangleright$ #1}
  \begin{algorithmic}[1]
    \Require{fidelity threshold $\mathcal{F}_\text{Th}$, Lipschitz bound $L_{\hat{H}}$, nominal parameter value $\mu_0$, fidelity-band tolerance $\eta$, domain $\Omega_\mu = [\underline{\mu},\overline{\mu}]$, maximum steps $K_{\max}$ per direction.}
    \For {$s\in\{-1,+1\}$} \Comment{$s=\pm1$: decreasing/increasing}
    \State{$k \gets 1$; $\nu \gets \mu_0$} \Comment{$k$: steps taken; $\nu$: working iterate}
    \State{$\mathrm{done} \gets \mathbf{false}$; $\mathrm{status}_{s} \gets \textsc{unset}$}
    \While{not $\mathrm{done}$}
    \State{$\nu^{+} \gets \nu + s(\mathcal{F}_{\nu}-\mathcal{F}_\text{Th})/L_{\hat{H}}$}
    \State{Clamp $\nu^{+}$ to $\Omega_\mu$}
    \State{Evaluate fidelity $\mathcal{F}_{\nu^{+}}$ at $\nu^{+}$}
    \If{numerical evaluation gives $\mathcal{F}_{\nu^{+}} < \mathcal{F}_\text{Th}$}
    \LineComment{Numerical safeguard only}
    \State{Bisect $[\nu, \nu^{+}]$ and}
    \State{\hspace*{1.2em} update $\nu^{+}$ such that $\mathcal{F}_{\nu^{+}}\ge\mathcal{F}_\text{Th}$}
    \EndIf
    \If{$\nu^{+} \in \partial \Omega_\mu$ and $\mathcal{F}_{\nu^{+}} - \mathcal{F}_\text{Th} \geq \eta$}
    \LineComment{Domain-truncated: margin $\geq$ distance to $\partial\Omega_\mu$}
    \State{$M_{s} \gets |\mu_0 - \nu^{+}|$}
    \State{$\mathrm{status}_{s} \gets \textsc{domain-trunc}$; $\mathrm{done} \gets \mathbf{true}$}
    \ElsIf{$0 \le \mathcal{F}_{\nu^{+}}-\mathcal{F}_\text{Th} < \eta$}
    \State{$M_{s} \gets \left| \mu_0 - \nu^{+} \right|$}
    \State{$\mathrm{status}_{s} \gets \textsc{eta-band}$; $\mathrm{done} \gets \mathbf{true}$}
    \ElsIf{$k \geq K_{\max}$}
    \LineComment{Certified lower bound only}
    \State{$M_{s} \gets \left| \mu_0 - \nu^{+} \right|$}
    \State{$\mathrm{status}_{s} \gets \textsc{iteration-limit}$; $\mathrm{done} \gets \mathbf{true}$}
    \Else
    \State{$\nu \gets \nu^{+}$; $k \gets k+1$}
    \EndIf
    \EndWhile
    \EndFor
    \State{\Return $\mathcal{M} = \min\{M_{-1},M_{+1}\}$ with the directional statuses $\mathrm{status}_{-1},\mathrm{status}_{+1}$}
  \end{algorithmic}
\end{algorithm}

\paragraph{One-dimensional iterative margin}
A single application of the Lipschitz radius is typically conservative~\cite{schirmer2024,oneil_2024_sensitivity_bounds}, so Algorithm~\ref{algorithm} repeats the certified update along each direction. Reevaluating the fidelity surplus $\mathcal{F}_\nu-\mathcal{F}_\text{Th}$ at successive certified points extends the safe interval while $L_{\hat{H}}=B_T C_{\hat{H}}$ stays fixed. In exact arithmetic,
\begin{equation}
  \nu^{+} = \nu \pm \frac{\mathcal{F}_{\nu}-\mathcal{F}_\text{Th}}{L_{\hat{H}}}
\end{equation}
cannot cross below the threshold by Theorem~\ref{thm:safe_radius}: exact threshold contact may place an iterate on the component boundary. Suppose the chosen direction has a finite threshold boundary, $0<L_{\hat{H}}<\infty$, and no earlier domain truncation. In exact arithmetic, with the finite stopping rules omitted, the iterates are monotone and bounded by that boundary, so they converge. An interior limit would have positive fidelity surplus and hence a positive limiting step, contradicting convergence. Thus the sequence approaches the boundary unless it reaches it exactly. Threshold contact may be tangential and need not start a violating interval. The iteration stops when $0\le\mathcal{F}_{\nu^{+}}-\mathcal{F}_\text{Th}<\eta$, when the boundary of $\Omega_\mu$ is reached while remaining safe, or when the step limit $K_{\max}$ is reached. Bisection is used only as a safeguard when finite-precision or approximate fidelity evaluation reports $\mathcal{F}_{\nu^{+}}<\mathcal{F}_\text{Th}$.

The directional value $M_s$ returned by a finite run is a certified lower bound on the distance to a threshold boundary, when one exists in direction $s$, and hence also on the distance to a strictly violating perturbation. This bound need not equal either distance. The quantity $\eta$ is a tolerance in fidelity space: it does not provide a two-sided parameter-space bracket or a prescribed relative accuracy in $M_s$. The status $\mathrm{status}_s$ records \emph{which} stopping rule was reached; it is not an error estimate for the margin. A \textsc{domain-truncated} result certifies only that the margin is at least the distance to $\partial\Omega_\mu$. We return
\begin{equation}
  \mathcal{M} = \min\{M_{-1},M_{+1}\},
\end{equation}
which certifies $\mathcal{F}_\mu\ge\mathcal{F}_\text{Th}$ on the symmetric interval $|\mu-\mu_0|\le\mathcal{M}$ within $\Omega_\mu$, including any threshold endpoint reached. Disconnected safe intervals beyond the nominal component are outside this certificate. Numerical results use floating-point fidelity evaluations without rigorous rounding-error bounds.

\begin{figure*}[!t]
  \centering
  \begin{minipage}{0.32\linewidth}
    \centering
    \includegraphics[width=\linewidth]{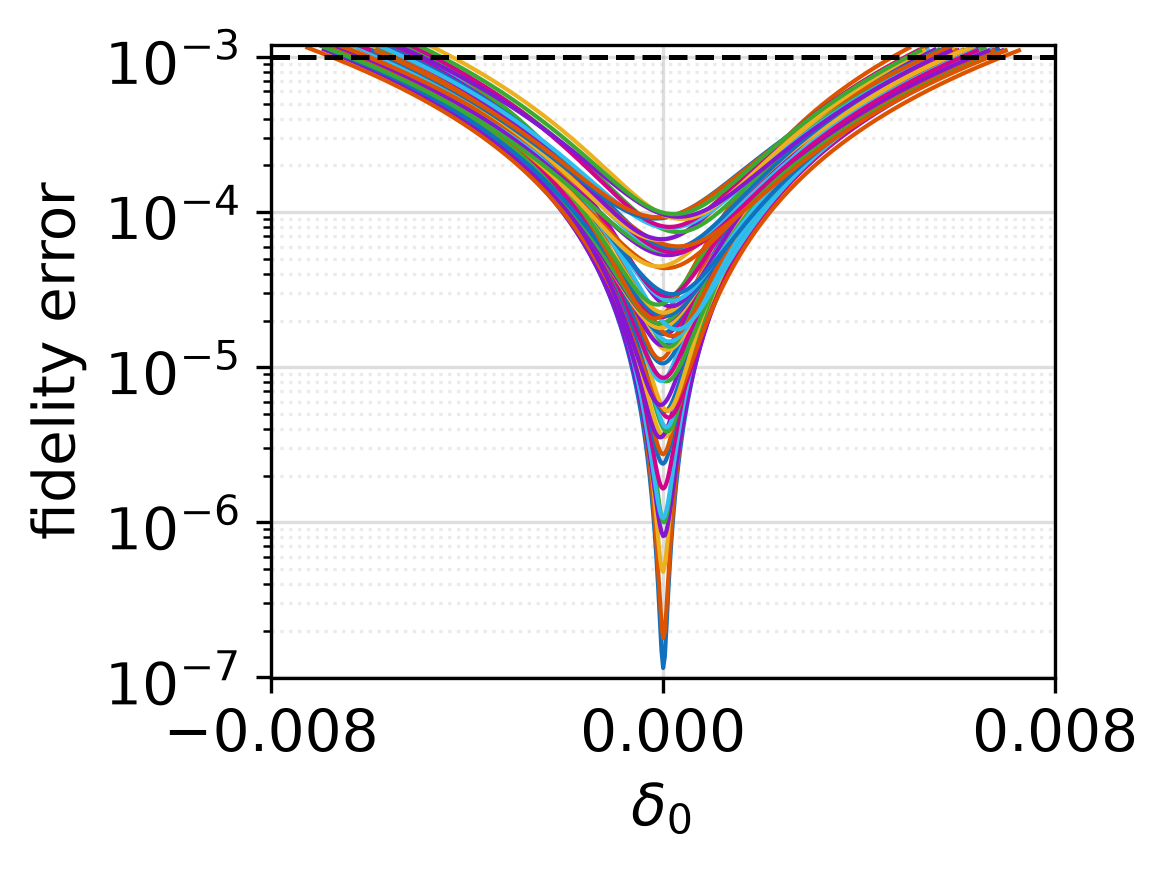}
    \\ \footnotesize (a) Perturbation $H_0$.
  \end{minipage}
  \hfil
  \begin{minipage}{0.32\linewidth}
    \centering
    \includegraphics[width=\linewidth]{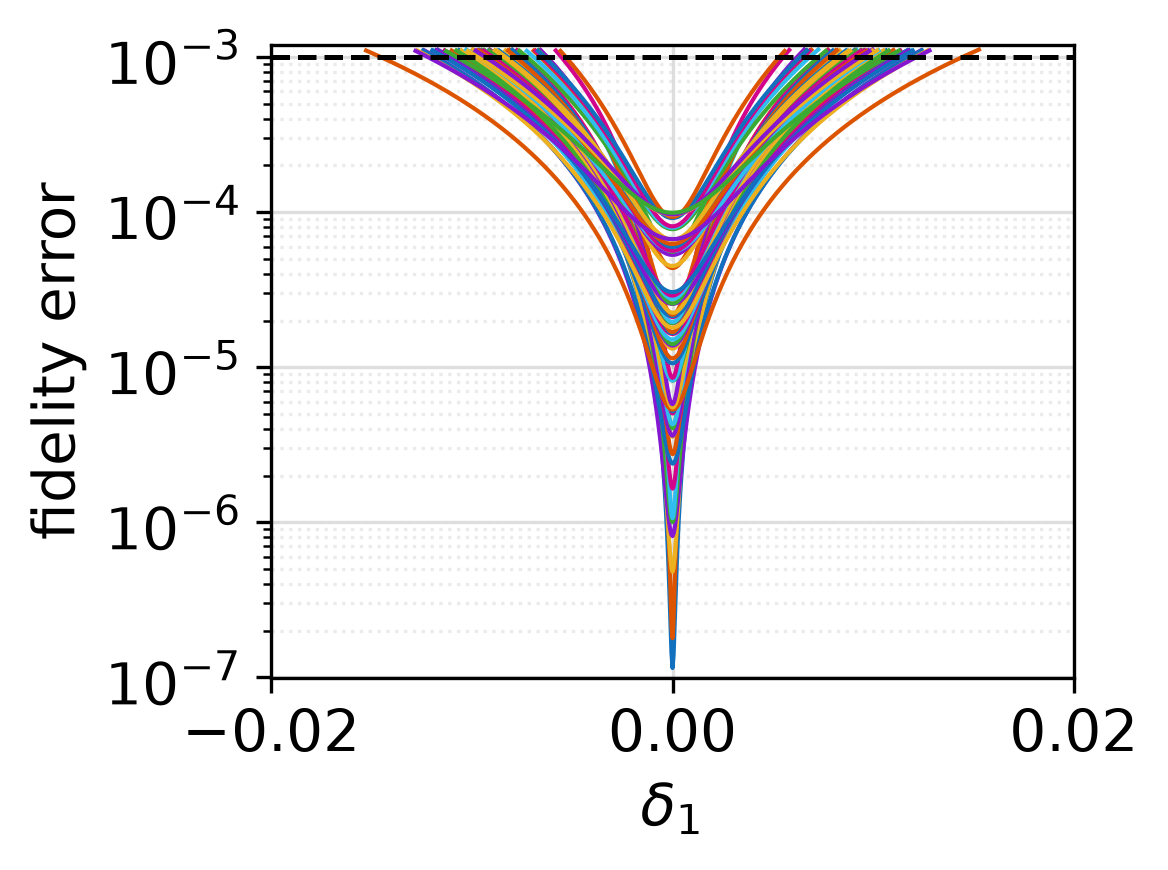}
    \\ \footnotesize (b) Perturbation $H_1$.
  \end{minipage}
  \hfil
  \begin{minipage}{0.32\linewidth}
    \centering
    \includegraphics[width=\linewidth]{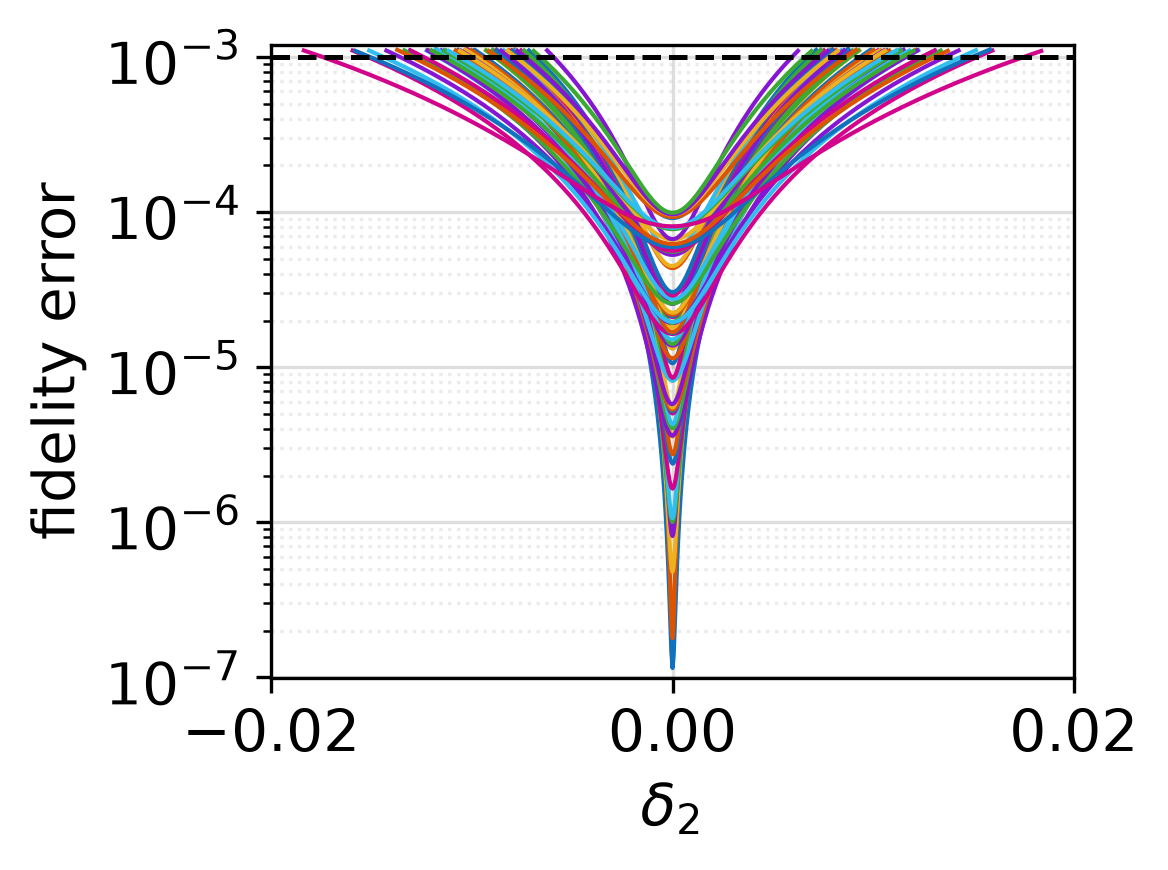}
    \\ \footnotesize (c) Perturbation $H_2$.
  \end{minipage}
  \caption{Fidelity error versus multiplicative perturbation strength $\delta_j$ for the $61$ controllers: (a)~$H_0$, (b)~$H_1$, and (c)~$H_2$. Each curve represents one controller, using the same color cycle across the three panels. The black dashed line marks the fidelity-error threshold $10^{-3}$.}\label{fig:fid_err_vs_perturbation_strength}
\end{figure*}

\begin{figure}[!t]
  \centering
  \includegraphics[width=0.95\linewidth]{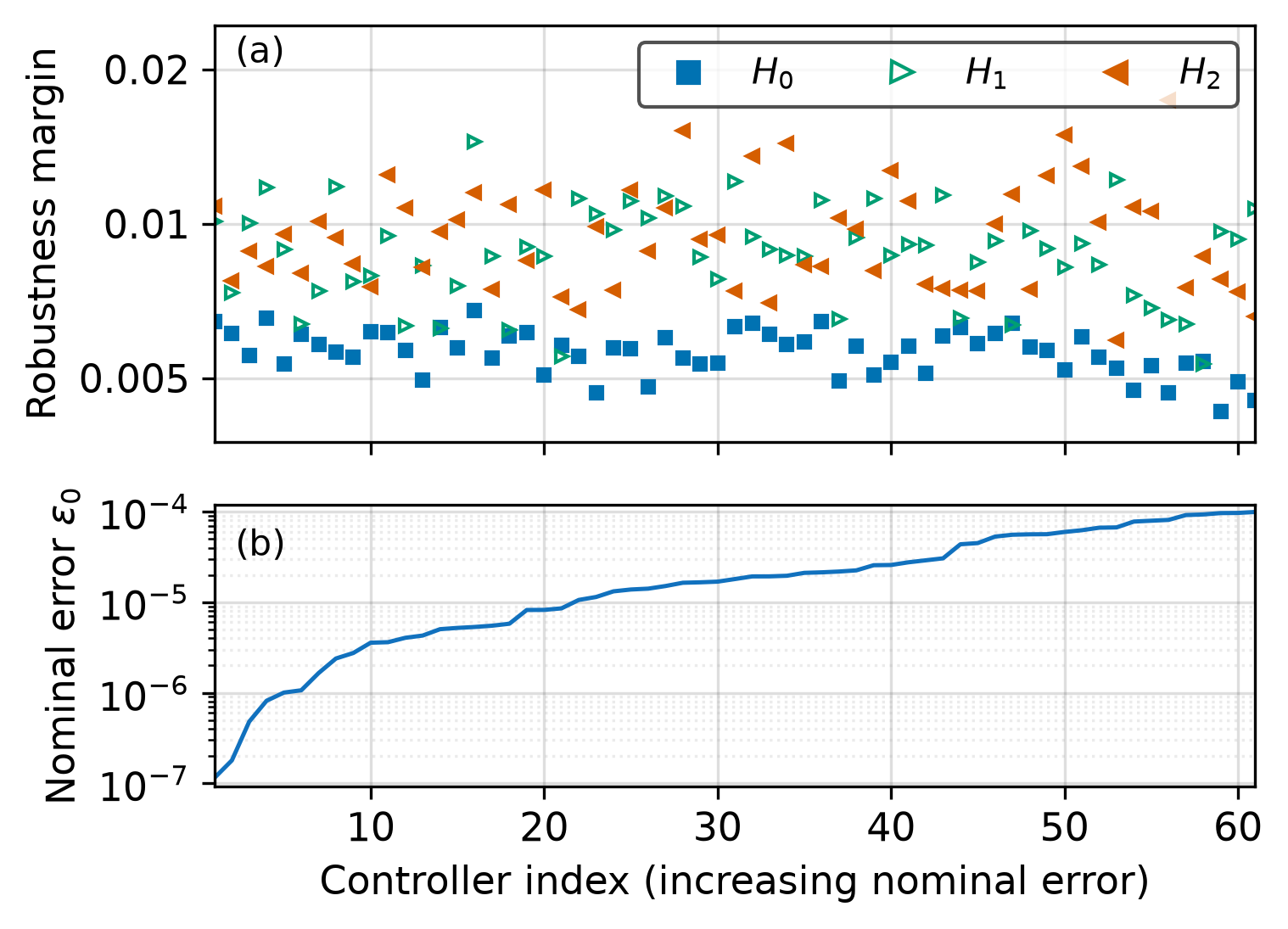}
  \caption{(a)~Robustness margins for $H_0$ (blue squares), $H_1$ (open green right-pointing triangles), and $H_2$ (filled orange left-pointing triangles). (b)~Nominal fidelity error $\varepsilon_0$.}\label{fig:margins_vs_fid_err}
\end{figure}

\begin{figure}[!t]
  \centering
  \includegraphics[width=0.95\linewidth]{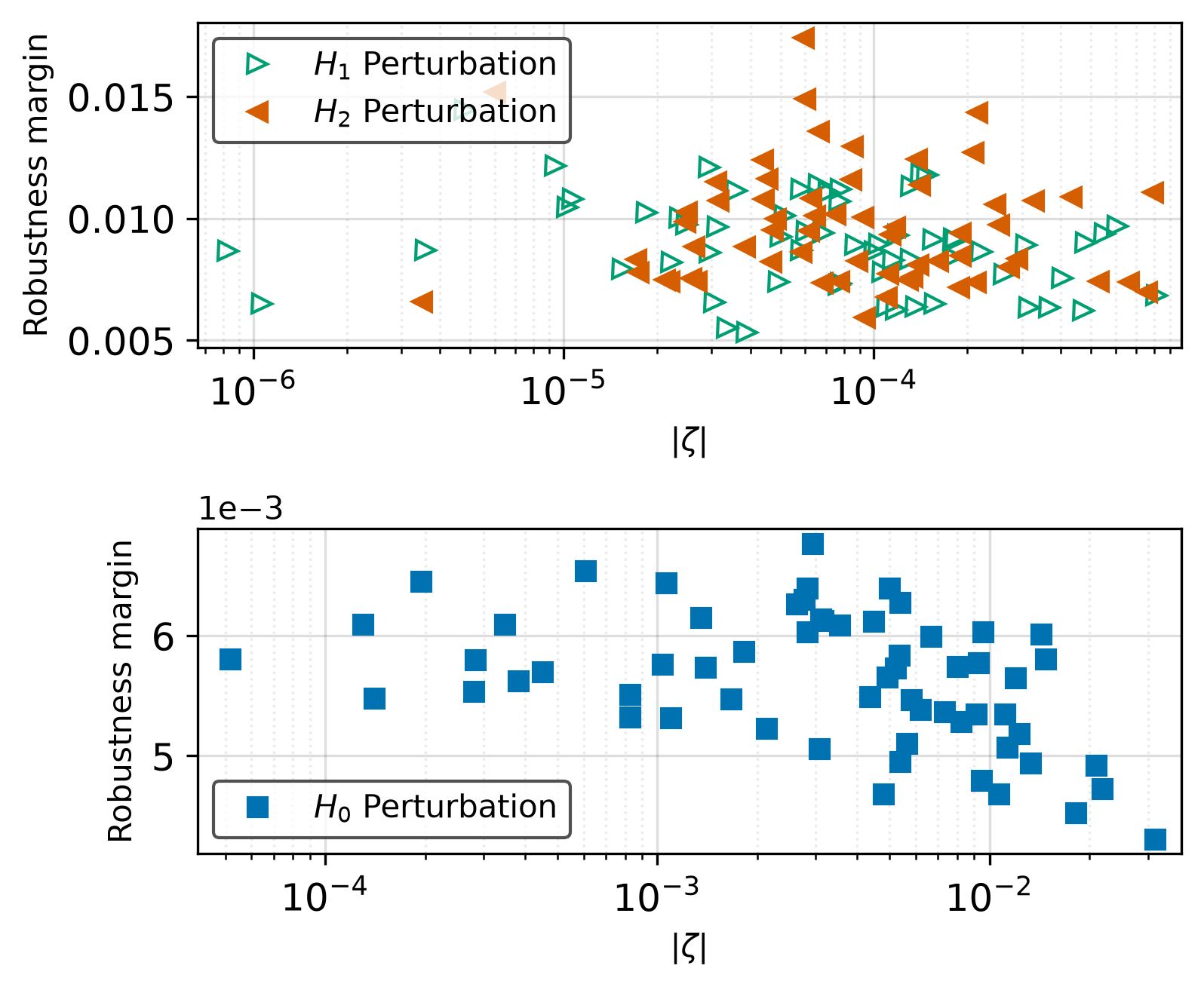}
  \caption{Robustness margins versus nominal differential-sensitivity magnitude $|\zeta_j|$ for $H_1$, $H_2$ (top) and $H_0$ (bottom).}\label{fig:margins_vs_diff_sens}
\end{figure}

\begin{table}[!t]
  \caption{Descriptive correlations of nominal error $\varepsilon_0$, margins $M_j:=M_{H_j}$, and nominal sensitivity magnitudes $|\zeta_j|$, $\zeta_j:=\zeta_{H_j}$ from~\eqref{eq:sens}, at $\mathcal{F}_\text{Th}=1-10^{-3}$. Reversing the parameter coordinate preserves $M_j$ but changes the sign of $\zeta_j$, so we compare magnitudes. Pearson $r$ measures linear association. Spearman $\rho$ measures monotone rank association. Cells above main diagonal: Pearson $r$; below: Spearman $\rho$.}\label{tab:correlations}
  \centering
\begin{tabular}{@{}lccccccc@{}}
  \toprule
  & $\varepsilon_0$ & $M_0$ & $M_1$ & $M_2$ & $|\zeta_0|$ & $|\zeta_1|$ & $|\zeta_2|$ \\
  \midrule
  $\varepsilon_0$ & $1.00$ & $-0.45$ & $-0.16$ & $0.02$ & $0.69$ & $0.06$ & $-0.06$ \\
  $M_0$ & $-0.38$ & $1.00$ & $0.12$ & $0.00$ & $-0.58$ & $0.13$ & $0.23$ \\
  $M_1$ & $-0.04$ & $0.08$ & $1.00$ & $-0.10$ & $-0.06$ & $-0.25$ & $-0.06$ \\
  $M_2$ & $-0.04$ & $0.04$ & $-0.13$ & $1.00$ & $0.10$ & $-0.07$ & $-0.12$ \\
  $|\zeta_0|$ & $0.70$ & $-0.46$ & $-0.06$ & $0.12$ & $1.00$ & $-0.07$ & $-0.10$ \\
  $|\zeta_1|$ & $0.01$ & $0.17$ & $-0.26$ & $-0.07$ & $0.06$ & $1.00$ & $0.43$ \\
  $|\zeta_2|$ & $-0.04$ & $0.30$ & $-0.08$ & $-0.06$ & $-0.16$ & $0.48$ & $1.00$ \\
  \bottomrule
\end{tabular}

\end{table}

\section{Case Study: Gate Optimization}\label{sec:example}

We consider controllers optimized for gate fidelity in a three-spin Heisenberg chain~\cite{Floether_2012}. Control acts only on the initial spin, making optimization more challenging than with full spin addressability.  Using time units $J^{-1}$ and energy units $\hbar J$, with $J$ the angular coupling frequency, the dimensionless drift $H_0$ and control Hamiltonians $H_1,H_2$ (all in $\mathbb{C}^{8\times8}$) are
\begin{align}
  \begin{split}\label{eq:case_study_hamiltonians}
    &H_0 = \frac{1}{2} \sum_{\ell = 1}^{2} \left( \sigma_{x}^{(\ell)}\sigma_{x}^{(\ell +1)} + \sigma_{y}^{(\ell)} \sigma_{y}^{(\ell+1)} + \sigma_{z}^{(\ell)}\sigma_{z}^{(\ell+1)} \right), \\
    &H_{1} = 2 \sigma_{x}^{(1)}, \quad H_{2} = 2 \sigma_{y}^{(1)},
  \end{split}
\end{align}
where $\sigma_q$, $q\in\{x,y,z\}$, are the Pauli spin operators. Here $\sigma_q^{(\ell)}$ is the three-fold tensor product with $\sigma_q$ in the $\ell$th position and the $2\times 2$ identity matrix $I_2$ in the others. All controllers share one fixed randomly drawn target unitary $U_f$, with $U(0)=I_8$. As in~\cite{oneil_2024_sensitivity_bounds}, the dimensionless gate duration is $t_f=15$ (physical duration $15/J$), with $\tau=32$ time steps.

We use $100$ previously optimized controllers and retain $61$ with nominal error $\varepsilon_0\le10^{-4}$. Margins use $\mathcal F_\text{Th}=1-10^{-3}$, $\eta=10^{-6}$, $\mu_0=0$, $\Omega_\mu=\mathbb R$ and at most $\QRMAuditKMax$ steps per direction. All $\QRMAuditRays$ rays stopped in the fidelity band, with $\QRMAuditGuardRays$ safeguard activations. Continuation used a median of $\QRMAuditContinuationMedian$ fidelity evaluations per controller/structure pair (both directions). The toolbox provides the data and code~\cite{qrobustness}.

Post hoc refinement used geometric probes to find parameter values with fidelity strictly below the threshold, followed by bisection. A new safe endpoint is accepted only when safe radii certify its connection to the current endpoint, possibly through intermediate certified steps. All directions reached relative bracket width $\QRMAuditTolerance$, with a median of $\QRMAuditRefinementMedian$ additional evaluations per controller/structure pair. The smaller distance from $\mu_0$ to the two resulting below-threshold points gives the upper bound $\mathcal M_{\mathrm{upper}}$. Relative to the original margin, $(\mathcal M_{\mathrm{upper}}-\mathcal M)/\mathcal M\le\QRMAuditRelativeWidth$; refinement changed the spread factors by less than $0.1\%$. These brackets describe numerical tightness for this ensemble.

In the comparisons below, $\zeta_j = \partial\mathcal{F}/\partial\delta_j|_{\delta_j=0}$ is evaluated once, at the nominal controller. Algorithm~\ref{algorithm} neither uses nor recomputes $\zeta_j$ along the continuation path. Each certified step uses only the uniform constant $L_{\hat{H}_j}$ and the fidelity at the current recentering point.

As discussed in Section~\ref{sec:robustness-margins}, we take the uncertainty structures $\{\hat{H}_j\}_{j=0}^{2}$ to be the (unnormalized) Hamiltonians $H_0$, $H_1$, and $H_2$ in~\eqref{eq:case_study_hamiltonians}, matching the convention used for $C_{\hat{H}}$ there. Reported margins are therefore in the same units as the corresponding multiplicative perturbation of those Hamiltonians. Each scalar margin varies one parameter while all others remain nominal. Their Cartesian product is not, in general, a certified region for simultaneous uncertainty.

Fig.~\ref{fig:fid_err_vs_perturbation_strength} shows fidelity-error trajectories under $H_0$, $H_1$, $H_2$ perturbations. Table~\ref{tab:correlations} and Figs.~\ref{fig:margins_vs_fid_err}--\ref{fig:margins_vs_diff_sens} show that the relationships between nominal fidelity error, nominal differential-sensitivity magnitude, and finite robustness margins are structure-dependent. For the drift structure, $|\zeta_0|$ is positively associated with $\varepsilon_0$ and negatively associated with $M_0$ under both Pearson and Spearman measures. For the control structures, the corresponding associations are weak, especially for $H_2$. Thus, the data show a clear association for drift uncertainty but no consistent ranking of the control-structure margins from nominal sensitivity alone. The finite margins therefore add information about off-nominal behavior that is not fully captured by nominal fidelity or local sensitivity. The $H_0$ margins occupy a narrower range than the $H_1,H_2$ margins in this example, but the reported quantities are structure-specific multiplicative perturbation margins, so the comparison should not be interpreted as a universal physical ordering.

\vbox{These comparisons are conditional on the selected ensemble, its single target unitary, and $\mathcal{F}_\text{Th}=1-10^{-3}$. Selection by $\varepsilon_0\le 10^{-4}$ narrows the observed nominal-error range and can affect these correlations, which may not apply over a wider fidelity range. Persistence of the controller rankings and correlations at nearby thresholds has not been assessed.}

\section{Conclusion}\label{sec:conclusion}
We developed a structure-specific scalar fidelity-threshold margin for finite-time gate control. A bound on the sensitivity of the (trace-amplitude) gate fidelity is used to derive a threshold-dependent certified local radius for a physical Hamiltonian parameter, while recentered continuation advances this certificate toward the first boundary of the nominal connected safe component. At $\mathcal{F}_\text{Th}=1-10^{-3}$, the lower bounds obtained for the $\QRMAuditControllers$ three-qubit controllers span factors of $\QRMAuditDriftSpread$ for the drift structure and $\QRMAuditControlOneSpread$--$\QRMAuditControlTwoSpread$ for the control structures. The drift margin is appreciably anticorrelated with the nominal sensitivity magnitude, whereas the control-structure margins show only weak associations with their corresponding nominal sensitivities in this selected ensemble. Extensions to simultaneous parameters, arbitrary within-gate trajectories, and margin-aware synthesis are left to subsequent work.


\begin{thebibliography}{10}
\providecommand{\url}[1]{#1}
\csname url@samestyle\endcsname
\providecommand{\newblock}{\relax}
\providecommand{\bibinfo}[2]{#2}
\providecommand{\BIBentrySTDinterwordspacing}{\spaceskip=0pt\relax}
\providecommand{\BIBentryALTinterwordstretchfactor}{4}
\providecommand{\BIBentryALTinterwordspacing}{\spaceskip=\fontdimen2\font plus
\BIBentryALTinterwordstretchfactor\fontdimen3\font minus \fontdimen4\font\relax}
\providecommand{\BIBforeignlanguage}[2]{{%
\expandafter\ifx\csname l@#1\endcsname\relax
\typeout{** WARNING: IEEEtran.bst: No hyphenation pattern has been}%
\typeout{** loaded for the language `#1'. Using the pattern for}%
\typeout{** the default language instead.}%
\else
\language=\csname l@#1\endcsname
\fi
#2}}
\providecommand{\BIBdecl}{\relax}
\BIBdecl

\bibitem{Koch_2022}
C.~P. Koch, U.~Boscain, T.~Calarco, G.~Dirr, S.~Filipp, S.~J. Glaser, R.~Kosloff, S.~Montangero, T.~Schulte-Herbr{\"u}ggen, D.~Sugny, and F.~K. Wilhelm, ``Quantum optimal control in quantum technologies. {Strategic} report on current status, visions and goals for research in {Europe},'' \emph{{EPJ} Quantum Technol.}, vol.~9, no.~1, p.~19, 2022.

\bibitem{Petersen2013}
I.~R. Petersen, ``Robustness issues in quantum control,'' in \emph{Encyclopedia of Systems and Control}, 2nd~ed., J.~Baillieul and T.~Samad, Eds.\hskip 1em plus 0.5em minus 0.4em\relax Cham: Springer International Publishing, 2021, pp. 1986--1991.

\bibitem{Automatica}
C.~A. Weidner, E.~A. Reed, J.~Monroe, B.~Sheller, S.~P. O'Neil, E.~Maas, E.~A. Jonckheere, F.~C. Langbein, and S.~G. Schirmer, ``Robust quantum control in closed and open systems: Theory and practice,'' \emph{Automatica}, vol. 172, p. 111987, 2025.

\bibitem{Doyle1982}
J.~C. Doyle, ``Analysis of feedback systems with structured uncertainties,'' \emph{IEE Proceedings D (Control Theory and Applications)}, vol. 129, no.~6, pp. 242--250, Nov. 1982.

\bibitem{Zhou}
K.~Zhou and J.~C. Doyle, \emph{Essentials of Robust Control}.\hskip 1em plus 0.5em minus 0.4em\relax Upper Saddle River, NJ: Prentice Hall, 1998.

\bibitem{schirmer2024}
S.~P. O'Neil, C.~A. Weidner, E.~A. Jonckheere, F.~C. Langbein, and S.~G. Schirmer, ``Robustness of dynamic quantum control: Differential sensitivity bounds,'' \emph{{AVS} Quantum Sci.}, vol.~6, no.~3, p. 032001, 2024.

\bibitem{oneil_2024_sensitivity_bounds}
S.~P. O'Neil, E.~A. Jonckheere, and S.~Schirmer, ``Sensitivity bounds for quantum control and time-domain performance guarantees,'' \emph{{IEEE} Control Syst. Lett.}, vol.~8, pp. 169--174, 2024.

\bibitem{lidar2008}
D.~A. Lidar, P.~Zanardi, and K.~Khodjasteh, ``Distance bounds on quantum dynamics,'' \emph{Phys. Rev. A}, vol.~78, no.~1, p. 012308, 2008.

\bibitem{berberich2024}
J.~Berberich, D.~Fink, and C.~Holm, ``Robustness of quantum algorithms against coherent control errors,'' \emph{Phys. Rev. A}, vol. 109, no.~1, p. 012417, 2024.

\bibitem{berberich2025}
J.~Berberich, T.~Fellner, R.~L. Kosut, and C.~Holm, ``Robustness of quantum algorithms: Worst-case fidelity bounds and implications for design,'' arXiv preprint arXiv:2509.08481, 2025, revised Jan. 2026.

\bibitem{kosut2025}
R.~L. Kosut, D.~A. Lidar, and H.~Rabitz, ``A fundamental bound for robust quantum gate control,'' arXiv preprint arXiv:2507.01215, 2025.

\bibitem{kiely2024}
P.~M. Poggi, G.~{De Chiara}, S.~Campbell, and A.~Kiely, ``Universally robust quantum control,'' \emph{Phys. Rev. Lett.}, vol. 132, no.~19, p. 193801, 2024.

\bibitem{zou2025}
Z.-J. Chen, H.~Huang, L.~Sun, Q.-X. Jie, J.~Zhou, Z.~Hua, Y.~Xu, W.~Wang, G.-C. Guo, C.-L. Zou, L.~Sun, and X.-B. Zou, ``Robust and optimal control of open quantum systems,'' \emph{Sci. Adv.}, vol.~11, no.~9, p. eadr0875, 2025.

\bibitem{KHANEJA_2005}
N.~Khaneja, T.~Reiss, C.~Kehlet, T.~Schulte-Herbr{\"u}ggen, and S.~J. Glaser, ``Optimal control of coupled spin dynamics: Design of {NMR} pulse sequences by gradient ascent algorithms,'' \emph{J. Magn. Reson.}, vol. 172, no.~2, pp. 296--305, 2005.

\bibitem{oneil_geometric}
S.~P. O'Neil, E.~A. Jonckheere, and S.~Schirmer, ``Geometric interpretation of sensitivity to structured uncertainties in spintronic networks,'' \emph{{IEEE} Control Syst. Lett.}, vol.~9, pp. 192--197, 2025.

\bibitem{oneil_2024_log_sens}
S.~O'Neil, S.~G. Schirmer, F.~C. Langbein, C.~A. Weidner, and E.~A. Jonckheere, ``Time-domain sensitivity of the tracking error,'' \emph{{IEEE} Trans. Autom. Control}, vol.~69, no.~4, pp. 2340--2351, 2024.

\bibitem{Floether_2012}
F.~F. Floether, P.~de~Fouqui{\`e}res, and S.~G. Schirmer, ``Robust quantum gates for open systems via optimal control: {Markovian} versus non-{Markovian} dynamics,'' \emph{New J. Phys.}, vol.~14, no.~7, p. 073023, 2012.

\bibitem{qrobustness}
F.~C. Langbein, S.~P. O'Neil, S.~Schirmer, C.~A. Weidner, and E.~A. Jonckheere, ``Quantum robustness margins,'' Zenodo, 2026, ver.\ 1.0.3, doi: \href{https://doi.org/10.5281/zenodo.23712877}{10.5281/zenodo.23712877}.

\end{thebibliography}

\end{document}